\documentclass[11pt]{amsart}

\usepackage[colorlinks=true, pdfstartview=FitV, linkcolor=blue,
citecolor=blue]{hyperref}

\usepackage{amssymb,amsmath}
\usepackage{a4wide}
\usepackage{xcolor}
\newtheorem{theorem}{Theorem}
\newtheorem{coro}[theorem]{Corollary}

\newcommand{\ts}{\hspace{0.5pt}}

\newcommand{\RR}{\mathbb{R}\ts}
\newcommand{\NN}{\mathbb{N}}
\newcommand{\Sph}{\mathbb{S}}

\newcommand{\vG}{\varGamma}

\newcommand{\dd}{\,\mathrm{d}}
\newcommand{\ee}{\ts\mathrm{e}}
\newcommand{\ii}{\ts\mathrm{i}\ts}

\newcommand{\myfrac}[2]{\frac{\raisebox{-2pt}{$#1$}}
  {\raisebox{0.5pt}{$#2$}}}

\newcommand{\eqdef}{= \raisebox{0.2pt}{$\scriptstyle{:}$}}

\makeatletter
\renewcommand{\@captionfont}{\small}
\makeatother

\DeclareMathOperator{\vol}{vol}

\begin{document}
	
\title[A note on measures whose diffraction is concentrated on a
single sphere]{A note on measures whose diffraction\\[2mm] is
  concentrated on a single sphere}
	
\author{Michael Baake}
\address{Fakult\"at f\"ur Mathematik, Universit\"at Bielefeld,
	\newline \indent Postfach 100131, 33501 Bielefeld, Germany}
\email{$\{$mbaake,jmazac$\}$@math.uni-bielefeld.de}

\author{Emily R.~Korfanty}
\address{Department of Mathematical and Statistical Sciences,
  University of Alberta, \newline \indent
  Edmonton, AB, T6G 2G1, Canada}
\email{ekorfant@ualberta.ca}

\author{Jan Maz\'{a}\v{c}}

\subjclass{42B10}	
\keywords{Spherical diffraction, radial symmetry, Bessel functions}

\begin{abstract}
  Is there a translation-bounded measure whose diffraction is
  spherically symmetric and concentrated on a single sphere? This
  note constructively answers this question of Strungaru in the
  affirmative.
\end{abstract}

\maketitle
 
When studying the long-range order of a given structure, which can be
a point set, a function, or a translation-bounded measure, diffraction
analysis provides insight into the nature of the order. Diffraction
originates in physics and is a standard tool for understanding
crystals and, later, quasicrystals. As such, it is a well-established
method of structural analysis.

The mathematical theory of diffraction is related to signal
analysis, with early results (in particular on the autocorrelation
of functions) going back to Wiener \cite[Sec.~3]{Wiener}, which were
later extended by others; see also \cite[Def.~2.5]{Ben}. Further
background and references for Wiener's diagram can be found in
\cite[Ch.~9]{TAO}.  For aperiodic structures, this approach was
extended to translation-bounded measures by Hof \cite{Hof}, who
also defined the diffraction as the Fourier transform of the
autocorrelation, as we will briefly recall for bounded functions.

Consider a complex-valued, bounded function $g$ on $\RR^d$ that is
locally integrable, and define its \emph{natural Patterson function} 
$\eta$ via convolution as
\begin{equation}\label{eq:auto-def}
  \eta (x) \, =  \lim_{R\to\infty} 
  \frac{\bigl(g^{}_R * \widetilde{g^{}_R}\ts\bigr)(x)}
  {\vol (B^{}_R )} \ts ,
\end{equation}
provided the limit exists. Here, $B^{}_R$ denotes the closed ball of
radius $R$ around $0$ and $g^{}_{R}$ is the restriction of $g$ to
$B^{}_R$, where, for any function $h$, one has
$\widetilde{h} (x) \mathrel{\mathop:}= \overline{h(-x)}$. For
simplicity, the natural Patterson function is called the (natural)
\emph{autocorrelation} from now on, because it is a special case of an
autocorrelation measure.  If the autocorrelation exists, it does not
matter whether one takes balls or cubes, or other centred van Hove
sequences. For a spherical setting, using balls in the Euclidean norm
is most convenient.

When the autocorrelation exists, one can employ
\cite[Lem.~1.2]{Martin} together with the volume formula for
$B^{}_R$ to rewrite $\eta$ as
\[
  \eta (x) \, = \lim_{R\to\infty} \frac{\bigl(g^{}_R *
    \widetilde{g}\bigr)(x)}{\vol (B^{}_R )} \, = \lim_{R\to\infty}
  \myfrac{\vG \bigl(\frac{d}{2} + 1 \bigr)}{\pi^{\frac{d}{2}} \ts R^d}
  \int_{B^{}_R} g (y) \, \overline{g (x-y)} \dd y \ts ,
\]
where $\vG$ denotes the Gamma function, with
$\vG(x+1) = x \ts \vG(x)$, and values
$\vG\bigl(\tfrac{1}{2}\bigr)=\sqrt{\pi}$ and $\vG(1)=\vG(2)=1$. This
version of $\eta$ is often convenient for computations, especially in
the context of spherical symmetry.  By construction, the
autocorrelation $\eta$ of $g$ is a positive definite function, and is
thus Fourier transformable (in the sense of tempered
distributions). Its Fourier transform is a positive measure (by the
Bochner--Schwartz theorem), which is known as the natural
\emph{diffraction measure} of $g$; we refer to \cite[Ch.~9]{TAO} and
references therein for background. For the Fourier
transform of integrable functions on $\RR^d$, we use
\[
    \widehat{f} (x) \, = \int_{\RR^d} \ee^{-2 \pi \ii x y} f(y) \dd y \ts ,
\]
which is extended to finite and then translation-bounded measures in
the standard way. Here, it suffices to think in terms of tempered
distributions, which avoids some subtleties of the Fourier analysis of
unbounded measures; compare \cite[Ch.~8]{TAO} for more.

In the theory of aperiodic order, diffraction analysis is a powerful
tool for understanding the long-range order of aperiodic tilings
\cite{TAO,Moody}.  In the context of tilings with statistical circular
symmetry (such as the pinwheel tiling, whose diffraction is still an
open problem \cite{BFG}), the following question was asked by
Strungaru.

\vspace{0.7mm}

\noindent
\textit{Is there a planar structure (say, a bounded function, or a
  translation-bounded measure) whose diffraction pattern is uniformly
  distributed and concentrated on a single circle?}

\vspace{0.7mm}

Below, we show that the natural suspect, the spherical wave
$g(x) = \ee^{2\pi \ii r \|x\|}$, with fixed $r>0$, provides an
affirmative answer in the plane and, in fact, in all dimensions.
Here, the radius $r$ takes the role of the wave number (or
vector) $k$ in plane waves of the form $\ee^{2 \pi \ii k x}$.  Note
that our notion of a spherical wave should be distinguished from the
fundamental solution of the radial wave equation from
electrodynamics or optics.  Let us briefly note that a related
analysis of circular cosine functions and their Fourier spectrum was
presented in \cite{Amidror}; however, the setting and objectives
differ from those considered here.

In what follows, we establish that the autocorrelation of
such a function in $d$ dimensions exists and is given by
\begin{equation}\label{eq:autocor}
  \vG \Big( \myfrac{d}{2}\Big)\,
  \frac{J^{}_{\frac{d}{2}-1}\bigl( 2 \pi r \| x \|\bigr)}
  {\bigl( \pi r \| x \| \bigr)^{\frac{d}{2} - 1}} \ts , 
\end{equation}
where $J^{}_{\nu}$ refers to the standard Bessel function of integral
or half-integral order; we refer to \cite[Chs.~9 and 10]{AS} for
details. Each $J^{}_{\nu}$ is an entire function, with series
expansion
\[
  J_{\nu} (z) \, = \, \Bigl( \myfrac{z}{2}\Bigr)^{\nu}
  \sum_{m=0}^{\infty} \myfrac{(-1)^{m}}{m \ts ! \, \vG(\nu+m+1)}
  \, \Bigl(\myfrac{z}{2}\Bigr)^{2m}.
\]

By standard integration in spherical coordinates, if follows that the
Fourier transform of~$\mu^{}_{r}$, the uniform probability measure on
the sphere of radius $r$ in $\RR^d$, satisfies
\begin{equation}\label{eq:sphere-Bessel}
    \widehat{\mu^{}_{r}} (x) \, = \int_{\RR^d} \! \ee^{-2 \pi \ii x y}
    \dd \mu^{}_{r} (y) \, = \, \vG \Big( \myfrac{d}{2}\Big)\,
    \frac{J^{}_{\frac{d}{2}-1}\bigl( 2 \pi r \| x \| \bigr)}
    {\bigl( \pi r \| x \| \bigr)^{\frac{d}{2} - 1}} \ts .
\end{equation}
This also holds in the measure-theoretic sense by considering
functions as Radon--Nikodym densities relative to Lebesgue measure;
compare \cite[Rem.~30]{SpSt}.  Consequently, if we show that the
autocorrelation of the spherical wave is of the form
\eqref{eq:autocor}, the Fourier inversion formula gives the desired
claim.  In what follows, we provide a short proof of this claim, which
highlights all important steps.

\begin{theorem}\label{thm:sphere-d}
  For any fixed\/ $r>0$, the natural autocorrelation\/
  $\eta^{}_{r}$ of the spherical wave\/ $\ee^{2 \pi \ii r \| x \|}$ 
  with\/ $x\in\RR^d$ exists, and is given by
\[
    \eta^{}_{r} (x) \, =  \,   \vG \Big( \myfrac{d}{2}\Big)\,
    \frac{J^{}_{\frac{d}{2}-1}\bigl( 2 \pi r \| x \| \bigr)}
    {\bigl( \pi r \| x \| \bigr)^{\frac{d}{2} - 1}} \ts .
\]
For\/ $r\to 0^{+}$, this converges to the constant function\/ $1$,
so\/ $\eta^{}_{0} (x) \equiv 1$ and\/ $\widehat{\ts\eta^{}_{0}\ts} 
= \delta^{}_{0}$.
\end{theorem}

\begin{proof} 
Since $\eta^{}_{r} (x) = \eta^{}_{1} (r x)$ for
any $r>0$ and $x\in\RR^d$, it suffices to consider
$\eta = \eta^{}_{1}$. Thus, we are interested in the
autocorrelation defined by
\[
    \eta(x) \, = \lim_{R\to \infty} \myfrac{1}{\vol(B^{}_{R})}
    \int^{}_{B^{}_{R}} \! \ee^{2\pi\ii \|y\|} \ts \ee^{-2\pi\ii \|x-y\|}
     \dd y \ts .
\]
We rewrite the integral in spherical coordinates and choose
$x = (s,\, 0\ts, \, \dots \, 0)$, as the resulting function $\eta$ is
radially symmetric (due to rotation invariance of $B^{}_{R} (0)$
and Euclidean volume). This gives a representation via a double
integral as 
\[
  \eta(x) \, = \lim_{R\to \infty}
  \myfrac{\Theta^{}_d}{\vol(B^{}_{R})}\, \int_0^R \!
  \int^{\pi}_0 \! \rho^{d-1}
  \sin(\theta)^{d-2} \ee^{2\pi\ii \rho} \ee^{-2\ts\pi\ii
    \sqrt{\rho^2+s^2-2 \rho s \cos(\theta)} }  
   \dd \theta  \dd \rho \ts ,
\]
where,  for any fixed $R$, the order of the integrals does not
matter due to Fubini's theorem; compare 
\cite[Thms.~{\!}VI.8.4 and VI.8.7]{Lang}. Let us also note that the
function of $s$ defined by the double integral is \emph{symmetric}
in $s$ because replacing $s$ by $-s$ can be undone by the
substitution $\theta \mapsto \pi - \theta$. Now, a standard
computation gives
\[
  \Theta^{}_d \, = \int^{\pi}_0 \sin(\theta^{}_2)^{d-3}
  \dd\ts\theta^{}_2 \, \cdots \int^{\pi}_0 \!
  \sin(\theta^{}_{d-2}) \dd\ts\theta^{}_{d-2}
  \int^{2\pi}_0 \!\!\! \dd\ts\theta^{}_{d-1} \, = \,
  \myfrac{2 \ts \pi^{\frac{d-1}{2}}}{\vG\left(\frac{d-1}{2}\right)}\ts ,
\]
which means that we are dealing with 
\begin{align*}
  \eta(x) \, & = \, h(s) \, =  \lim_{R\to \infty} h^{}_{R}(s) \\[2mm]
  \, &  = \,  \myfrac{2\ts\vG\left(\frac{d+2}{2}\right)}{\sqrt{\pi}
       \ts \vG\left(\frac{d-1}{2}\right)} \, \lim_{R\to \infty}
       \myfrac{1}{R^d} \int_0^R \! \int^{\pi}_0 \!
       \rho^{d-1} \sin(\theta)^{d-2}
       \ee^{2\pi\ii \rho} \underbrace{\ee^{-2\pi\ii
       \sqrt{\rho^2+s^2-2 \rho s \cos(\theta)}}}_{\eqdef H(s)}
       \dd \theta  \dd \rho  \ts .
\end{align*}
The existence of the limit, for any fixed $s$, is a simple
consequence of the fact that the absolute value of the (continuous
and differentiable)  integrand is bounded by $\rho^{d-1}$, and the
absolute value of the double integral (including the prefactor) then
by $\frac{\pi}{d}$. As this estimate is
independent of $s$, the convergence of
$h^{}_{R} (s)$ as $R \to \infty$ is uniform in $s$, and the limit
function $h$ is continuous. In fact, the functions are also
differentiable to an arbitrary order, which we now want to utilise.

In order to apply Leibniz' rule for differentiation under the
integral, see \cite[Lemma~VIII.2.2]{Lang}, we need to tackle the
singularities of the derivatives of $H$ at $\rho = s$ for
$\theta =0$. To this end, observe that, for any fixed $s$, we can
start the radial integration at some $R_0 > s$ without affecting the
limit, which can be viewed as a simple regularisation. If we
consider $s \in [-M, M]$ for some fixed $M>0$, we may choose
$R_0 = 2 M$ and replace the functions $h^{}_{R}$ by truncated
versions $h^{}_{R,R_0}$, where the radial integration starts at
$R_0$. These functions still converge to $h$, and uniformly so on
$[-M,M]$. In particular, differentiation commutes with the limit.
What is more, we can now apply Leibniz' rule for
derivatives to any of the $h^{}_{R,R_0}$, calculate them via the
corresponding derivatives of $H$ under the integral, and use the
fact that they converge uniformly on $[-M,M]$, namely to the
corresponding derivatives of $h$, where the radial cutoff aviods
all singularities.  

Our strategy now employs the Taylor series of $h$ around  zero,
which can initially be considered as a potentially asymptotic
expansion, though we shall see that its radius of convergence
is positive. The required derivatives of $H$, for any $m\in\NN$
and $\rho \geqslant R_0 > M$, are given by 
\[
  H^{(m)}(s) \, = \, (-2\pi\ii)^m  \ee^{-2\pi\ii
    \sqrt{\rho^2+s^2-2\ts \rho s\cos(\theta)}}
  \left(\myfrac{\rho \ts \cos(\theta) - s}
    {\sqrt{\rho^2+s^2-2 \rho s\cos(\theta)}\ts }\right)^{\! m}
  \! + \, \mathcal{O}^{}_{\theta} (\rho^{-1})  \ts ,
\]
where $\mathcal{O}^{}_{\theta}$ is meant to indicate that the
constant of the error term depends on $\theta$, but is integrable
over the interval $[0,\pi]$ because $\rho$ stays clear of all
singularities.  For $s=0$, we have
\[
  H^{(m)}(0) \,= \, (-2\pi\ii)^m \ts
     \ee^{-2\pi\ii \rho} \cos(\theta)^m
     + \, \mathcal{O}^{}_{\theta} (\rho^{-1}) \ts ,
\] 
where the error term does not lead to any trouble because
it is multiplied with $\rho^{d-1}$ in the radial integral, with
$d\geqslant 2$.  Then, after performing some spherical integration,
it becomes clear that one can replace $R_0$ again by $0$ without
any diﬀerence to the limit as $R \to \infty$. With this
replacement, the remaining integration step reads
\begin{align*}
  h^{(m)}(0) & \, = \, \myfrac{2 \ts (-2 \pi \ii)^{m}
               \vG\left(\frac{d+2}{2}
               \right)}{\sqrt{\pi} \ts \vG\left(\frac{d-1}{2}\right)}
               \lim_{R\to \infty}  \myfrac{1}{R^d}
            \int_0^R  \rho^{d-1} \! \int_{0}^{\pi} \sin(\theta)^{d-2}
               \cos(\theta)^m   \dd \theta \dd \rho 
               + \mathcal{O}(R^{-1})  \\[2mm]
             & \, = \, \myfrac{(-2\pi\ii)^m}{\sqrt{\pi}} \,
               \myfrac{\vG\left(\frac{d}{2}\right)}
               {\vG\left(\frac{d-1}{2}\right)}
               \int^{\pi}_0 \! \sin(\theta)^{d-2} \cos(\theta)^m
               \dd \theta  \\[2mm]
             & \, = \, \begin{cases}
       \myfrac{(-2\pi\ii)^m}{\sqrt{\pi}}\,
       \myfrac{\vG\left(\frac{d}{2}\right)\vG
       \left(\frac{m+1}{2}\right)}{\vG\left(\frac{d+m}{2}\right)} \ts ,
          & \mbox{if}\ m \ \mbox{is even,} \\
            0 \ts , & \mbox{otherwise}. \end{cases}
\end{align*}

Consequently, the Taylor series reads
\begin{align*}
h(s) \, & =  \sum_{m=0}^{\infty} \myfrac{ h^{(m)}(0)}{m!} \, s^m 
       \, = \sum_{m=0}^{\infty} \myfrac{(-1)^m}{(2m)!}\,
       \myfrac{\vG\left(\frac{d}{2}\right)\vG
         \left(\frac{2m+1}{2}\right)}{\sqrt{\pi} \, \vG
         \left(\frac{d}{2}+m\right)}\, (2\pi s)^{2m} \\[2mm]
     & = \, \vG\left(\myfrac{d}{2}\right) \sum_{m=0}^{\infty}
       \myfrac{(-1)^m}{\vG(m+1)\vG\left(\frac{d}{2}+m\right)}
       \, (\pi s)^{2m} \, = \,
       \myfrac{\vG\left(\frac{d}{2}\right)}
         {(\pi s)^{\frac{d-1}{2}}} \, J^{}_{\frac{d}{2}-1} (2\pi s),
\end{align*}
where we used Legendre's duplication formula for the Gamma function
\cite[p.~256]{AS},
\[
  \vG(2z) \, = \, \myfrac{1}{\sqrt{\pi}}
  \bigl(2^{2z-1} \vG(z)\,  \vG(z+\tfrac{1}{2}) \bigr)\ts ,
\]
with $2z = 2m+1$. At this point, we also see that the series
for the symmetric function $h$ actually has infinite radius of
convergence, as expected. 

This completes the claim, as the remaining cases can be obtained by
rescaling the argument, while the situation for $r=0$ is elementary.
\end{proof}

As explained above, we thus also have the following result. 

\begin{coro}\label{coro:wave}
  For any fixed\/ $r>0$, the natural diffraction measure of the
  spherical wave $f(x) = \ee^{2\pi \ii r \|x\|}$ is\/ $\mu^{}_{r}$,
  the probability measure for the uniform distribution on the
  sphere\/ $r \ts \Sph^{d-1} = \partial B_{r}(0)$, with the obvious
  extension to the limiting case\/ $r=0$. \qed
\end{coro}

Our above argument is explicit and, with hindsight, quite
straightforward. However, it now raises several natural questions
about combinations of such waves and the superposition of spheres;
compare \cite{Amidror} for related problems. Since explicit
calculations quickly get out of hand, the setting asks for a more
systematic approach via singular measures, suitable orthogonality
notions, and a careful analysis of the conditions under which the
autocorrelation is well defined and bounded. One route along
these lines, which also contains an alternative proof of our result
in Corollary~\ref{coro:wave}, has now been taken in \cite{KM}. 

\bigskip

\section*{Acknowledgements}

It is our pleasure to thank Nicolae Strungaru for valuable
discussions. We thank an anonymous referee for his constructive
comments, which helped us to improve the presentation.  This work was
supported by the German Research Council (Deutsche
Forschungsgemeinschaft, DFG) under CRC1283 at Bielefeld University
(project ID 317210226). E.R.K.\ acknowledges travel support from grant
2024-04853 of the Natural Sciences and Engineering Research Council of
Canada (NSERC). 

\end{document}